\documentclass[10pt, conference, letterpaper]{IEEEtran}
\makeatletter
\def\ps@headings{%
\def\@oddhead{\mbox{}\scriptsize\rightmark \hfil \thepage}%
\def\@evenhead{\scriptsize\thepage \hfil \leftmark\mbox{}}%
\def\@oddfoot{}%
\def\@evenfoot{}}
\makeatother
\usepackage{geometry}
\usepackage{epsfig}
\usepackage{amssymb,amsmath,amsfonts}
\usepackage{amsthm}
\usepackage{mathrsfs}
\usepackage{graphicx}
\usepackage{multirow}
\usepackage{color}
\usepackage{adjustbox}

\usepackage{epsfig} 
\usepackage{mathrsfs}
\usepackage{graphicx}
\usepackage{multirow}
\usepackage{mathtools}
\usepackage{caption}
\usepackage{subcaption}
\usepackage{booktabs}
\usepackage{caption}
\usepackage[font=footnotesize]{caption}

\usepackage{framed}
\usepackage{pbox}
\usepackage{tikz}
\usepackage{graphicx,adjustbox}
\usetikzlibrary{decorations.pathreplacing}
\usepackage{float}
\usepackage{anyfontsize}

\usepackage[%linesnumbered,
ruled,vlined]{algorithm2e}
\SetAlgorithmName{Protocol}{List of protocols}

\usepackage{dsfont}

\newcommand{\pr}[1]{\mathrm{Pr}\left(#1\right)} %Probability
\newcommand{\ev}[1]{\mathbb{E}\{#1\}}  %Expected value

\newcommand{\gf}[1]{\mathds{GF}\left(#1\right)} %Galois Field

\newcommand{\myvar}[1]{\mathds{V}\!\mathrm{ar}\left\{#1\right\}}  %Variance
\newcommand{\eps}{\epsilon}
\renewcommand{\varepsilon}{\epsilon}

\newcommand{\mc}[1]{\mathcal{#1}}

\newcommand{\TRASH}[1]{}

\providecommand{\eqref}[1]{(\ref{#1})}

\newcommand{\remove}[1]{}

\newcommand{\sd}{\mathbf{SD}}

\newtheorem{theorem}{Theorem}
\newtheorem{proposition}{Proposition}
\newtheorem{lemma}{Lemma}

\newtheorem{definition}{Definition}
\newtheorem{remark}{Remark}

\newtheorem{corollary}{Corollary}
\begin{document}

%Secure Message Transmission and 
%\title{A One-round Capacity Achieving Secret Key Agreement and its Applications}

% \title{Design and Implementation of a Key Agreement Protocol  for Long-term Security}

\title{A Capacity-achieving One-message Key Agreement With Finite Blocklength Analysis}

% %%% Single author, or several authors with same affiliation:
\author{%
  \IEEEauthorblockN{Setareh Sharifian,
                   Alireza Poostindouz, and
                    Reihaneh Safavi-Naini}
  \IEEEauthorblockA{University of Calgary\\
                    AB, Canada}
}

%Information-theoretic 
%Wiretap channels created by self jamming

\maketitle

%*****************************************************************************************%
%
% ABSTRACT
%
%*****************************************************************************************%

%\begin{thebibliography}{10} 
%\bibitem{Shannon}
%C. E. Shannon. A mathematical theory of communication. Bell System Technical Journal, 27:379?423, 623?656, 1948.

%\end{thebibliography} 

%comebackhere

%Infocom version: 
\begin{abstract}
Information-theoretic secret key agreement (SKA) protocols are a fundamental cryptographic primitive that 
are used to establish a shared secret key between two or more parties. %with infrm  do not rely on any computational assumption, and provide quantum-resistant  security.
In  a  two-party SKA in source model,  % the two-party source model for SKA, 
Alice and Bob have samples of two correlated variables, that are partially leaked to Eve, and their goal is to establish a shared secret key  by communicating over a   reliable public channel. Eve must have no information about the established key.
%
%  
%
%Alice and Bob use a public channel to send messages to agree on a secret key. The key must be information theoretically secure against Eve, who has access to the public messages and has some side information about Alice's and Bob's initial variables. 
In this paper, we study the problem of one-message secret key agreement where  the key is established by Alice sending a single message to Bob.
% In this problem, only one message from Alice to Bob is allowed to be used for key agreement. 
%
We propose a one-message SKA (OM-SKA) protocol,  prove that it achieves the one-way secret key capacity, and derive finite blocklength % for finite $ n$, we give two 
approximations of the %maximum 
achievable 
secret key length. 
%
% the achievable key length  is very close to its optimal value.  
% In  comparison   with a recently proposed  protocol that achieves the optimal rate   asymptotically  and for finite $n$, but requires up to $n$ communication rounds,  our protocol has the same asymptotic property but slightly smaller rate for finite $n$, while requiring a single message transmission.
%
% Applying these results  to the beacon generated correlation setting, gives explicit expressions  for achievable key length. 
% We propose  modifications to the protocol with
% detailed analysis and new security proofs  that results in a parametrized family of protocols, all with provable security but different efficiencies, 
% allowing the user to choose protocol parameters considering computational constraints.
%
We compare our results with existing %previous one-message SKA 
OM-SKAs and show the protocol has a unique combination of desirable properties.
%results in the literature, in terms of the secret key length, public communication cost, and computation complexities of Alice and Bob.
% and show that our bounds are tighter than other one-message SKA protocols. 
% and directions for future work. 
\end{abstract}

\section{Introduction}
	% through one-time-pad, we immediately have a secure message transmission protocol. 
	%\section{Secret Key Agreement: Source Model}
Key agreement is a fundamental problem in cryptography: Alice wants to share a secret key with  Bob  which  will be completely unknown to Eve.
In this paper, we consider information theoretic % we follow the information theoretic  approach to
 secret key agreement (SKA) that uses 
 physical layer assumptions %properties of the system (e.g. antenna diversity) and the environment (e.g. existence of multiple paths in a network) 
 to achieve security.
 %Using physical assumptions for secure message transmission was pioneered by 
 Wyner \cite{Wyner1975} pioneered the use of physical layer properties, in his case noise in the channel, for secure 
 message transmission.  The approach has found significant attention because of its application to many wireless communication settings, and 
 its information theoretic security which provides security even if a quantum computer exists.
% and extended by Csiszar and Korner \cite{Csiszar1978}.
 Information theoretic secret key agreement was first  proposed  %considered 
 by Maurer, and Ahlswede and Csisz{\'{a}}r independently \cite{Ahlswede1993,Maurer1993}.
 In the so-called 
 %In their model, called the 
 \textit{source model} of \cite{Ahlswede1993}, Alice and Bob  have %observe
  samples of correlated random variables (RVs) $X$ and $Y$, and want % seek
    to agree on a secret key by exchanging messages 
%They can communicati
 over a public  (authenticated and error free) channel that is visible to the eavesdropper Eve, who  has 
an initial %also some 
side information %represented by 
$Z$ about the correlated variables. 
The three %correlation between 
 variables  $X,Y$, and $Z$ have % is specified by 
 a  joint distribution $P_{XYZ}$ that is public.

 %Secret key agreement, in general, allows more freedom for physical layer restriction, as the goal is to agree on some secure random value. This is in contrast to information theoretic secure message transmission, in which a particular message should be securely transmitted. This means SKA might be possible even if secure message transmission is not. For instance, Maurer \cite{Maurer1993} considered a physical layer scenario in which secure message transmission is not possible between two parties, and showed that the parties still can agree on a secret key  with the use of public discussion.
 
In this paper, we study  {\em one-message secret-key agreement}  (OM-SKA) where the key is established by Alice sending a single message to Bob.
The problem is   first studied by \cite{Holenstein2006} and  \cite{holenstein2006strengthening},  and later, OM-SKA constructions based on polar codes where proposed in \cite{Renes2013} and \cite{Chou2015a}.
The problem is important in practice because it avoids interaction between Alice and Bob that would require stateful protocols with vulnerabilities in implementation.  It is also interesting theoretically due to \cite{Holenstein2006} as is related to circuit polarization and immunization of public-key encryption, or it can be used for oblivious transfer \cite{wullschleger2007oblivious}.

\remove{

 (e.g. theoretically as where   %for short under the source model with two distinctive properties, namely (i)
only one-way communication from Alice to Bob is allowed and (ii) Alice is only allowed to send \textit{one message} to Bob for the purpose of key agreement. This problem was first explored in \cite{Holenstein2006} and  \cite{holenstein2006strengthening}. 
One-message key agreement protocols are preferred over interactive alternatives, specially by small devices with power constrains, since cost of interaction for secret key agreement can become excessive for conventional key sizes. Moreover, when physical properties of the environment are varying (e.g. in Ad hoc networks), the environment may change during the interaction or transmission of more than one message. 
As a result, the initial physical assumptions may not hold and the security is not guaranteed anymore.   OM-SKA has other interesting applications in  cryptography namely for circuit polarization and immunization of public-key encryption which has been proposed and explored by Holenstein and Renner \cite{Holenstein2006}. It  also has applications for oblivious transfer as noted in \cite{wullschleger2007oblivious}. Study of these applications is beyond the scope of this paper.
}

Efficiency of an SKA is measured by the length $\ell$ of the established secret key. % that  is established. %hey generate. 
When  Alice,  Bob and Eve have $n$  independent samples  
of their random variables, denoted by $(X^n, Y^n,Z^n) $, the {\em rate of the secret key} is given by $\ell/n$. 
The {\em secret key capacity}  associated to a distribution $P_{XYZ}$  is  the highest achievable  key rate when ${n\to\infty}$.

\remove{
One important performance measure of a source model in which Alice,  Bob and Eve have their respective  $n$ samples of the random variables, $(X^n, Y^n,Z^n) $, is in terms of 
% the {\em secret key  rate} (achievable key length divided by $n$), and
% and 
the highest achievable  key rate (achievable key length divided by $n$) % that can be achieved  asymptotically  (
when ${n\to\infty}$,  %in the setting,  is given by
i.e., the {\em secret key capacity} of the  model. 
}
%%
%
%
%
 %%
% which is commonly  referred to as {\em source} (of information).
The secret key (SK) capacity $C_s(X,Y|Z)$ of % for
 a general  distribution is % has been  
 a long-standing open problem. %However if
When the three variables  form a Markov Chain ${X-Y-Z}$,
%the source model is called \emph{degraded} and
%(i.e., 
%$P_{XYZ}P_Y = P_{XY}P_{YZ}$
% $P(Z|XY)=P(Z|Y)$), 
   % It was  proved  \cite{Ahlswede1993}  that
   the secret key capacity  %of the source model  
   is given by,  
	$C_s(X,Y|Z) = I(X;Y|Z)$  \cite{Ahlswede1993}. 
%This relation holds in many settings of interest (see Section~\ref{sec:applications} for an example of such setting).
% Protocols that achieve this capacity for the special case of binary symmetric channels are given in \cite{Maurer1993,Ahlswede1993}.

The secret key capacity of a general distribution when the public communication channel is one-way (i.e., only Alice sending to Bob),
%For a general distribution (not necessarily satisfying $X-Y-Z$), if we restrict the model so that Alice and Bob can only use one-way public communication from Alice to Bob,
is called {\em  the one-way secret key (OW-SK) capacity}, and is denoted by $C_s^{ow}(X,Y|Z)$.
This capacity, by definition, is  a lower bound to the SK capacity, that is  $C_s^{ow}(X,Y|Z)\leq C_s(X,Y|Z)$ \cite{Ahlswede1993}.

%
%One-message SKA protocols are of special interest, since they are non-interactive and use only one-way public communication between Alice and Bob. 
%Throughout the paper we assume that the single public message is from Alice to Bob.

In a real life deployment of SKA protocols, $n$,  the number of available initial samples to each party, is finite and the efficiency of a protocol is measured by $\ell$. % and so an important performance measure is t
%the  {\em length $\ell$ of the established key}, or the corresponding rate $\ell/n$, rather than  the SK capacity. Therefore, from this practical view point, deriving finite-length approximations of the final secret key length is an important problem. 
In  \cite{Hayashi2016}, bounds on secret key  length in  finite blocklength regime are established using higher order approximations.
The given lower bound of \cite{Hayashi2016}, however, cannot be used for OM-SKA because the proposed SKA protocol is interactive and uses many ($\mc O(n)$) rounds of public communication. %  the limitation of one-way communication. %to obtain a general lower bound on the length

The first explicit constructions of %explicit 
OM-SKA protocols are given 
in \cite{Holenstein2006,holenstein2006strengthening}.   
The construction of \cite{Holenstein2006} uses  %using 
the concatenation of a random linear code with a Reed-Solomon code to allow Bob to recover the variable of Alice, and use that to derive a shared key.
The protocol achieves the OW-SK capacity.
There is no explicit finite blocklength analysis of this construction in \cite{Holenstein2006}, although as shown in Proposition~\ref{prop:1}, the work of \cite{holenstein2006strengthening} can be extended to obtain a lower bound on the key length.
% and a lower bound on the achievable finite-length key can be derived (See Proposition~\ref{prop:1}). 
% However, the rate of protocol is positive only when $n$, the number of samples, is larger than a high threshold. Other OM-SKA protocols are based on polar codes given by
 Renes, Renner and Sutter % et al.
  \cite{Renes2013},  and Chou, Bloch and Abbe %et al. 
  \cite{Chou2015a} proposed constructions of OM-SKA  using polar codes. These constructions are capacity achieving, but  their finite blocklength analysis, which is not discussed in \cite{Renes2013,Chou2015a}, will be directly related to the finite blocklength analysis of the underlying polar codes.
  \remove{
  do not have finite length
   Achievable finite-length lower bounds are not given for these protocols and can be subject to independent future works. 
}
%The ``one-message secret key agreement'' term has been used in \cite{ho}

\vspace{1mm} 
\noindent
{\bf  Our Work.}
%In this work, we first 
We propose a OM-SKA protocol, analyze its security and  %one-message SKA 
%protocol. We analyse our SKA construction and 
derive a finite blocklength lower bound for the key length.
The construction (and hence the bound) holds for the more general case that  the $n$ samples that are held by the parties are independent, but  could be drawn from different distributions (this was called ``independent experiments'' in \cite{Renner2005}).
The construction achieves the OW-SK capacity and provides a finite blocklength lower bound for one-way SKA.
We compare the finite blocklength lower bound of the protocol with the  only other lower bounds that we derived  based on the constructions in \cite{holenstein2006strengthening}, %and extending their analysis,
and show its superior performance.
In particular, as we illustrate in a numerical example, the latter bound of \cite{Holenstein2006} becomes  positive only when $n$, the number of samples, is larger than a high threshold ($\sim 10^{9}$ bits), while our bound is positive for values of $n$ starting from $10^{3}$ bits (See Figure~\ref{fig:comp}).

%only other known (our derivation)  f
% maximum achievable secret key length. This bound holds even if the n-tuple distribution of the source is due to independent experiments (not necessarily IID).  
\remove{
This bound shows that our SKA protocol can achieve the one-way SK capacity. For IID distributions, our analysis of the SKA protocol gives a second more accurate (tighter)  finite blocklength lower bound.  
}
An important property of  our OM-SKA protocol is  that it gives an explicit finite blocklength upper bound on the required  public communication  for achieving the OW-SK capacity. The bound only depends on $n$, $P_{XYZ}$, and key's secrecy and reliability properties. 
The full comparison of our protocol with all the known OM-SKA is given in Table~\ref{tab:comp}.

%Our SKA proof also gives a finite-length upper on the minimum amount of public communication bits that is required for achieving the one-way SK capacity. Lastly, we compare our results with other existing OM-SKA protocols. In our comparison, we consider the finite key length, public communication cost, and computation complexities of Alice and Bob. 

\vspace{1mm} 
{\noindent 
\bf Related Works.} 
%Sudy of information theoretic key agreement is initiated by 
Maurer \cite{Maurer1993} and Ahlswede and Csisz{\'{a}}r \cite{Ahlswede1993} initiated the  study of information theoretic secret key agreement, and  derived lower and upper bounds on the secret key capacity of the source model.  
There has been many follow up works, most notably for our study, deriving  upper and lower bounds for the secret key length with  matching (i.e., optimum) second order approximations \cite{Hayashi2016},  assuming the Markov chain ${X-Y-Z}$ holds.

One-way secret key (OW-SK) capacity was introduced by Ahlswede and Csisz{\'{a}}r \cite{Ahlswede1993}, who derived an expression to calculate % was proved for the calculation of 
 OW-SK capacity. 
Holenstein and Renner  \cite{Holenstein2006} considered one-message  SKA protocols 
and gave constructions that achieve  OW-SK capacity. 
%In \cite{Holenstein2006} and later in \cite{holenstein2006strengthening} Holenstein and Renner  give two explicit OM-SKA protocols that can achieve the OW-SK Capacity. 
% We extended their analysis to obtain finite blocklength lower bounds for these protocols. 
Their constructions % are considered \emph{efficient}, as they  have 
has computational complexity of $\mc O(n^2)$ .  
In \cite{Renes2013} and \cite{Chou2015a}, two capacity achieving OM-SKA constructions using polar codes 
are given. Although, these SKA protocols require specific construction of a polar code (i.e., computation of the index sets) for the underlying distribution,  and can benefit from progress in polar code construction. These codes do not provide finite blocklength bounds for the key length. %???

\vspace{1mm} 
{\noindent 
\bf Organization.}
The background is given in Section~\ref{sec: pre} and in Section~\ref{sec:Source_Model} we discuss about the OW-SK Capacity. Our main results are in Section~\ref{sec:ska}.
%The proofs can be found in the full version of this work in \cite{full-v1}. 
We conclude the paper in Section~\ref{sec:comparison} by comparing our results with previous works.

\section{Background}\label{sec: pre}

	\subsection{Notations and definitions}
	
	 We denote random variables (RVs) with upper-case letters, (e.g., $X$), and their realizations with lower-case letters, (e.g., $x$). 
	% We use %math
	 Calligraphic %upper-case 
	 letters to show the alphabet  of a random variable (e.g., $\mc X$). 
	  The probability mass function (p.m.f) of an RV %random variable
	   $X$ is denoted by $P_X(x) = \pr{X=x}$.
	  %  For 
	  %  A sequence of $n$ random variables is denoted by % we use 

	Shannon entropy of an RV $X$ over the alphabet  $\mathcal{X}$ is, %given by,
	 $H(X) = - \sum_{{x}\in \mathcal{X}}P_X({x}) \log P_X({x})$.  
	 \remove{
	 The binary entropy is $h_2(p)=-p\log p -(1-p)\log(1-p)$. The  min-entropy  of $X$  is given by:
	$
	H_{\infty}(X)= -\log (\displaystyle \max    P_X({x}) ).
	$ Logarithms %in this work
	 are in base 2. 
\iffalse    
    \begin{align*}
    H(X) &= - \mathds{E}_{P_{X}} \log P_X({X}),\\
    &= - \sum_{{x}\in \mathcal{X}}P_X({x}) \log P_X({x}),
    \end{align*}
\fi    
	%that is a statistical average over the information of individual members of $x^n$. 
	  %on the other hand measures the minimum information contained in the random variable and 
	%denoted by $H_{\infty}(X)$ and 
%For two random variables $X\in\mc{X}$ and $Y\in\mc{Y}$ with
The joint  and conditional  distributions of two RVs $X$ and $Y$ 
are, $P_{XY}$ and 
$P_{X|Y}$, respectively. The {\em conditional entropy}  %$H(X|Y)$ 
is  defined as $
		H(X|Y)=- \mathds{E}_{P_{XY}} \log P_{X|Y}({X}|{Y}),$
	and the 
	}
	The average conditional min-entropy is %defined as  %\cite[Section 2.4]{Dodis2008}:
	%$$\tilde{H}_\infty(X|Y) = -log \mathbb{E}_{y\in\mc{Y}}\max_{x}Pr[X = x|Z = z].$$
	$\tilde{H}_\infty(X|Y) = -\log \mathds{E}_{P_Y}  \max_{x\in\mc{X}}P_{X|Y}(x|Y)$ \cite{Dodis2008}.
The maximum number of extractable random bits from a random variable is given by its \textit{smooth min-entropy}
\cite{Renner2004} defined as
$
H_{\infty}^{\epsilon}(X)= \displaystyle \max_{Y:SD(X,Y)\leq \epsilon} H_{\infty}(Y)
$.

The \emph{mutual information} between $X$ and $Y$ is $I(X;Y)= %\mathds{E}_{P_{XY}} i(X,Y) = 
H(X)-H(X|Y).$ 
The conditional mutual information between $X$ and $Y$ given $Z$ is
 $I(X;Y|Z) = %\mathds{E}_{P_{XYZ}} i(X,Y|Z) = 
 H(X|Z)-H(X|Y,Z).%\\&=H(Y|Z)-H(Y|X,Z),
$
%where $i(x,y|z) = \log \frac{P_{XY|Z}(x,y|z)}{P_{X|Z}(x|z) P_{Y|Z}(y|z) },$
%     \begin{align*}
%     i(x,y|z) = \log \frac{P_{XY|Z}(x,y|z)}{P_{X|Z}(x|z) P_{Y|Z}(y|z) },
%     \end{align*} 
%is the {\em conditional information density. }
% {Random variables $(X,Y,Z)$ satisfy Markov relation, denoted by  $X-Y-Z$, if $P_{XYZ}(x,y,z)P_Y(y) = P_{XY}(x,y)P_{YZ}(y,z)$. }%Note that w
% If  %the Markov relation
% $X-Y-Z$ then, %holds, we have
% $I(X;Y|Z) = H(X|Z) - H(X|Y).$

If for RVs $(X,Y,Z)$ the Markov relation $X-Y-Z$ holds, i.e., $P_{XYZ}(x,y,z)P_Y(y) = P_{XY}(x,y)P_{YZ}(y,z)$ then, 
$I(X;Y|Z) = H(X|Z) - H(X|Y).$

%DO YOU USE ALL ABOVE IN THIS PAPER- AND WITHOUT INCLUDING PROOF??

%%%%%%%%%%%%%%%%%%%%%%%%%%%%
%%%%% Info com ver %%%%%%%%%
%%%%%%%%%%%%%%%%%%%%%%%%%%%%

% The \emph{mutual information} between two random variables is $I(X;Y)= \mathds{E}_{P_{XY}} i(X,Y) = H(X)-H(X|Y),$ where $i(x,y) = \log \frac{P_{XY}(x,y)}{P_X(x) P_Y(y) },$
% %     \begin{align*}
% %     i(x,y) = \log \frac{P_{XY}(x,y)}{P_X(x) P_Y(y) },
% %     \end{align*} 
% is called the {\em  information density} \cite{Han2003}. The conditional mutual information of $X$ and $Y$ given $Z$ is
% %     \begin{align*}
% % 			I(X;Y|Z)&= \mathds{E}_{P_{XYZ}} i(X,Y|Z), \\
% %             &= H(X|Z)-H(X|Y,Z),\\&=H(Y|Z)-H(Y|X,Z),
% % 	\end{align*}
%  $
% 			I(X;Y|Z) = \mathds{E}_{P_{XYZ}} i(X,Y|Z) = H(X|Z)-H(X|Y,Z),%\\&=H(Y|Z)-H(Y|X,Z),
% $
%     where $i(x,y|z) = \log \frac{P_{XY|Z}(x,y|z)}{P_{X|Z}(x|z) P_{Y|Z}(y|z) },$
% %     \begin{align*}
% %     i(x,y|z) = \log \frac{P_{XY|Z}(x,y|z)}{P_{X|Z}(x|z) P_{Y|Z}(y|z) },
% %     \end{align*} 
% is the {\em conditional information density. }{\color{black}Random variables $(X,Y,Z)$,  are said to satisfy Markov relation, denoted by  $X-Y-Z$, if $P_{XYZ}(x,y,z) = P_X(x)P_{Y|X}(y|x)P_{Z|Y}(z|y)$. }%Note that w
% When %the Markov relation
% $X-Y-Z$ holds, we have $I(X;Y|Z) = H(X|Z) - H(X|Y).$

%%%%%%%%%%%%%%%
%%%%%%%%%%%%%%%
%%%%%%%%%%%%%%%
	%
The {\em statistical distance}  between  two RVs $X$ and $Y$ with % with respective 
p.m.f's $P_X$ and $P_Y$, defined over % and a common
a common alphabet $\mc W$, is %denoted by $\sd(X;Y)$  and %that
%is defined as,
given by, $\sd(X;Y) = \frac{1}{2} \sum_{w\in \mc W} \left\vert P_X(w) - P_Y(w) \right\vert.$
%
	%Min-entropy and statistical distance are related through the definition of \emph{smooth min-entropy}. % which is first presented by Renner and Wolf. 
	%\begin{definition} 
		%\cite{Renner2004}
		%(\textit{smooth min-entropy \cite{Renner2004}}) 
		%For a random variable $X$ with alphabet $\mathcal{X}$, 

\remove{
A distribution $P_{W^n}$ is IID if we have $P_{W^n}=\Pi_j P_{W_j}$ and $P_{W_j}=P_W~\forall j$. Given the expected value and variance of a distribution, \textit{Berry-Essen inequality} gives an accurate estimate of the probability of the corresponding IID distribution. We use this inequality to make our achievability results more accurate in Section~\ref{sec:ska}.

\begin{lemma}
    [Berry-Esseen Inequality \cite{berry1941accuracy,esseen1945fourier}] 
    \label{lem:Berry-Esseen}
Let $P_{W^n}$ be an IID distribution, then for any $-\infty<\alpha<\infty$
\begin{equation*}
    \left\vert \pr{\sum_{j=1}^{n} W_j \leq n\mu - \alpha\sqrt{\Delta n}} - Q(\alpha) \right\vert \leq \frac{3\rho}{\Delta^{3/2} \sqrt{n}} , 
\end{equation*}
where $\mu=\ev{W}, \Delta=\myvar{W}, \rho=\ev{|W-\mu|^3}$, and $Q(\cdot)$ is the tail probability of the standard Gaussian distribution.
\end{lemma}

}% end remove 

     \begin{lemma}
    [Berry-Esseen Inequality \cite{berry1941accuracy,esseen1945fourier}] 
    \label{lem:Berry-Esseen}
Let $P_{W^n}$ be an IID distribution, then for any $-\infty<\alpha<\infty$
\begin{equation*}
    \left\vert \pr{\sum_{j=1}^{n} W_j \leq n\mu - \alpha\sqrt{\Delta n}} - Q(\alpha) \right\vert \leq \frac{3\rho}{\Delta^{3/2} \sqrt{n}} , 
\end{equation*}
where $\mu=\ev{W}, \Delta=\myvar{W}, \rho=\ev{|W-\mu|^3}$, and $Q(\cdot)$ is the tail probability of the standard Gaussian distribution.
    \end{lemma}
\subsection{Universal Hash Functions}\label{subsec:hash}
% We recall the definition of \textit{pair}-wise universal hash functions \cite{Carter1977a}, the leftover-hash lemma \cite{Impagliazzo1989a} and its variations which will be used for proving reliability and security of the proposed key agreement protocol.

\begin{definition}[\text{2-Universal Hash Family (2-UHF)\cite{Carter1977a}}]
\label{def: UHF}
A family $\{h_s|s\in\mathcal{S}\}$ of functions $h_s:\mathcal{X}\rightarrow\mathcal{Y}$ is a 2-Universal Hash Family if for any $x\neq x^{'}$,
$
\mathrm{Pr}\{h_S(x)= h_S(x^{'})\}\leq \frac{1}{|\mathcal{Y}|},
$
where the probability is on  the uniform choices over $\mathcal{S}$. %of   the seed 
% $S$ % that is  % denotes a random seed chosen
%  %uniformly chosen 
%  from $\mathcal{S}$.
\end{definition}

\remove{
{An example of a 2-UHF %2-Universal Hash Family 
is the {\em multiplicative hash family} ${\mathcal{H}_{mult}^{n,t}:} \{0,1\}^n\times\{0,1\}^n \to \{0,1\}^t$ that, on seed $  s \in \{0,1\}^n\setminus0^n$ and input $x^n \in \{0,1\}^n$, outputs the first $t$ bits of $x^n\odot   s$. Here, $\odot$ is multiplication over $\gf {2^n}$. This hash function is denoted as $H_{  s}^{n,t}( {x})=x^n \odot   s|_t$.}
%From now on by a Universal Hash Function (UHF) we mean a pair-wise universal hash function.
}
The \textit{Leftover Hash Lemma} (LHL) \cite{Impagliazzo1989a} states that 2-UHFs extract randomness from a source with a lower bound on its min-entropy \cite{Impagliazzo1989a}.
 \iffalse
\begin{lemma} (LHL)
Assume that the a family $\{h_s|s\in\mathcal{S}\}$ of functions $h_s:\{0,1\}^n\rightarrow\{0,1\}^\ell$ is a 2-UHF family where $S$ is uniform over $\mc{S}$. Then 
$$\sd((S,h_S(X));(S,U_{\ell}))\leq \frac{1}{2}\sqrt{2^
{\ell-H_\infty(X)}},$$
where $U_\ell$ denotes the uniform probability distribution over $\{0,1\}^\ell$.
\end{lemma} 
 \fi
The 
%following 
average-case version of LHL is given in \cite{Dodis2008}.
\remove{
\begin{lemma}\label{lem: ALHL} Let $\{h_s|s\in \mc{S}\}$  be a 2-UHF with  $h_s:\{0,1\}^n\rightarrow\{0,1\}^\ell$. Let $X$ and $Z$ be random variables over $\{0,1\}^n$ and $\{0,1\}^{*}$, respectively with the joint distribution $P_{XZ}$.  Then 
$$
\sd((S,Z,h_S(X));(S,Z,U_\ell))\leq \frac{1}{2}\sqrt{2^{\ell-\tilde{H}_\infty(X|Z)}},
$$
where $U_\ell$ is the uniform probability distribution over $\{0,1\}^\ell$.
\end{lemma}
}
%\section{Related works and summary of our result}

\subsection{SKA in Source Model}

An Information-theoretic key agreement protocol, has %usually has
 two main steps \cite{helleseth1993secret}:  %. In the 
 %upon the parameters and operators they will use the next two phases. In the 
 {\em information reconciliation} where the goal is to arrive at
% Alice and Bob  is to arrive at 
a common string, and %If 
 {\em privacy amplification}  where the goal is to extract a secret key from the shared string. Sometimes an initiation phase is included in the protocol during which
  protocol parameters  and public values are  determined.
The following definitions for information-theoretic SKA are consistent with the corresponding ones in \cite{Holenstein2006}.
        \begin{definition}[See \cite{Holenstein2006,Tyagi2014,Hayashi2016}]\label{def:seckey}
            Let Alice and Bob be two parties with inputs $X$ and $Y$ and local independent randomness sources $U_A$, $U_B$ respectively. A random variable $K $ over  $\mc{K}$  is an $(\epsilon,\sigma)$-Secret Key (in short $(\epsilon,\sigma)$-SK), if there exists a  protocol with public communication $\boldsymbol F$, and two functions $K_A(X,U_A,\boldsymbol F)$ and $K_B(Y,U_B,\boldsymbol F)$,  
            that satisfy the following reliability and security properties: % defined as follows:
            \begin{align}
                \text{(reliability)}\quad& \pr{K_A=K_B=K}\geq 1-\epsilon, \\
                \text{(security)}\quad& \sd\left((K,{\boldsymbol{F}} ,Z);(U,{\boldsymbol{F}},Z)\right) \leq \sigma,\label{def:sec}
	        \end{align}
	        %\vspace{ -0.3em}
	        \noindent where $U$ is sampled  uniformly from $\mc{K}$, and $Z$ is a random variable corresponding to 
	        Eve's side information. 
	       % any leaked information about $X$ and $Y$.
	        %and $U_x$ and $U_y$ are local sources of independent uniform randomness.
        \end{definition}

       Efficiency of an SKA protocol is in terms of the secret key length that is obtained for a given set of variables.

    \begin{definition}[\cite{Hayashi2016}]
		%Given a source model $(X,Y,Z)$ and a 
		For a given source model $(X,Y,Z)$ with joint distribution $P_{XYZ}$, and  pair of reliability and secrecy parameters $(\epsilon,\sigma) \in [0,1)^2$, 
		 the highest achievable key length is denoted by  $S_{\epsilon,\sigma}(X,Y|Z)$, and  is defined
		% $S_{\epsilon,\sigma}(X,Y|Z)$ to be 
		by the supremum of the  key length $\log |\mc{K}|$  for all $(\epsilon,\sigma)$-SKA protocols $\Pi:$ %$K$
		%i.e., %That is
	%	$S_{\epsilon,\sigma}(X,Y|Z) = \sup \{ \log|\mc{K}| \mid \text{$K$ with alphabet $\mc{K}$ is an $(\eps,\sigma)$-SK for $(X,Y,Z)$ } \}$.
	\begin{align*}
			S_{\epsilon,\sigma}(X,Y|Z) = \sup\limits_{\Pi}  \{ \log|\mc{K}| \mid \text{$K$ is $(\eps,\sigma)$-SK for $(X,Y,Z)$}\}.
		\end{align*}
	\end{definition}

%\section{Two-parry SKA  %Two-terminal 
%in Source Model}
% for SKA}
\section{One-way Secret Key Capacity}
\label{sec:Source_Model}

%One-way and one-message secret key agreement traces back to 
Ahlswede and Csisz{\'{a}}r \cite{Ahlswede1993}  %who derived 
derived the ``forward key capacity'' (or what we call the one-way secret key capacity) of the source model.
%\vspace{ 0.5em}
 %
Let  $X^n = (X_1,\ldots,X_n)$, $Y^n = (Y_1,\ldots,Y_n)$   and $Z^n = (Z_1,\ldots,Z_n)$ denote $n$ IID samples of the distribution $P_{XYZ}$ that is publicly known, and consider a protocol family indexed by $n$, achieving a secret key  of  length $\ell(n)$.
   
The {\em  secret key rate} of  the protocol, indexed by $n$, is given by $R(n)= \ell (n)/n $, and the achievable rate of the family is given
by $R^*=\liminf_{n \rightarrow \infty} \ell  (n)/n.$  The secret key \textit{capacity} is the supremum of the achievable key rate 
of  all protocol families for the same setting. The following definition follows the definitions in \cite{Hayashi2016}.

%over all achievable secret key rates.
	   \remove{
	    the protocol family %an SKA protocol family that achieves key length $\ell(n)$  
	   is given by,  $R(n)= \ell (n)/n $, 
	 %  $R(n)= \frac{\max(\ell_{\eps,\sigma}^{\mathbf{\Pi}_{\mathrm{SKA}}})(n)}{n} $, 
	 and the (asymptotically) achievable  secret key rate of the protocol family 
	   % \footnote{Note that throughout whenever $\liminf_{n \rightarrow \infty} f(n) = \lim_{n \rightarrow \infty} f(n)$, we use the `$\lim$' notation instead of `$\liminf$'.} %of the protocol
	     is  given by %defined by,
	% $R^*=\liminf_{n \rightarrow \infty} \frac{\max(\ell_{\eps,\sigma}^{\mathbf{\Pi}_{\mathrm{SKA}}})(n)}{n}.$
	 $R^*=\liminf_{n \rightarrow \infty} \ell  (n)/n.$
	
	 The secret key \textit{capacity} is supremum over all achievable secret key rates.
	 }
\remove{    
	\begin{definition}{\color{black}}
		For any given source model $( X, Y, Z)$ with IID distribution 
		$P_{X^n Y^n Z^n}$, the secret key capacity is defined by
		\begin{equation}
			C_s( X, Y| Z) =  \sup_{\Pi_n} \liminf_{n \rightarrow \infty} \frac{1}{n} S_{\epsilon_n,\sigma_n}(X^n,Y^n|Z^n),
		\end{equation}
		where the supremum is over all protocols $\Pi_n$ with reliability and secrecy parameters $\epsilon_n,\sigma_n$, such that $\lim\limits_{n \rightarrow \infty} (\epsilon_n+\sigma_n) =0.$
	%	\begin{align*}
	%		\lim_{n \rightarrow \infty} \epsilon_n+\sigma_n =0.
%		\end{align*}
	\end{definition}

	In general, the secret key capacity is not known for any arbitrary source distribution \cite{ElGamal2011}. 
}

\begin{definition}[One-way Secret Key Capacity
\footnote{In some works including \cite{holenstein2006strengthening,Renner2005,Renes2013} the term ``capacity'' is reserved for %the case of
 physical channels,  and for 
%where a real channel exist and for referring to the described concept in
 source model only ``key rate'' is used.}]
		For a %any given 
		source model  with distribution 
		$P_{X^n Y^n Z^n}$, the secret key capacity $C_s^{ow}( X, Y| Z)$ is defined by
		\begin{equation}
			C_s^{ow}( X, Y| Z) =  \sup_{\Pi_n} \liminf_{n \rightarrow \infty} \frac{1}{n} S^{ow}_{\epsilon_n,\sigma_n}(X^n,Y^n|Z^n),
		\end{equation}
		where the supremum is over all protocols $\Pi_n$ with reliability and secrecy parameters $\epsilon_n,\sigma_n$, such that $\lim\limits_{n \rightarrow \infty} (\epsilon_n+\sigma_n) =0$ and $S^{ow}_{\epsilon_n,\sigma_n}(X^n,Y^n|Z^n)$ is the corresponding key length from the OW-SKA protocol ${\Pi_n}$.
\end{definition}

It was  
 shown  \cite[Theorem 1]{Ahlswede1993} that  %for one-way communication in the described source model, the 
 one-way secret key capacity is
given by the supremum
%\footnote{Although the results in \cite{Csiszar1978,Ahlswede1993} are given in terms of ``maximum'' (not ``supremum''), but since these results are non-constructive, we used ``supremum'' instead of ``maximum'' here.} 
of $H(U|Z,V)-H(U|Y,V)$, where the supremum is over all distributions $P_{UV}$ satisfying $V-U-X-YZ$. 
This result also follows from \cite[Corollary 2]{Csiszar1978}. 
Holenstein and Renner \cite[Theorem 3]{Holenstein2006} proved  %that same results hold when 
the supremum can be taken  over all distributions $P_{UV|X}$ satisfying $X-U-V$, and that  the supremum
can always be achieved by a OM-SKA protocol  \cite[Theorem 3.3]{holenstein2006strengthening}, and
so we have the following theorem for  OW-SKA  capacity.

\remove{

It was %has been
 shown in \cite[Theorem 1]{Ahlswede1993} that for one-way communication in the described source model, the one-way secret key capacity is
given by the supremum
%\footnote{Although the results in \cite{Csiszar1978,Ahlswede1993} are given in terms of ``maximum'' (not ``supremum''), but since these results are non-constructive, we used ``supremum'' instead of ``maximum'' here.} 
of $H(U|Z,V)-H(U|Y,V)$, where the supremum is over all distributions $P_{UV}$ satisfying $U-V-X-YZ$. This result also follows from \cite[Corollary 2]{Csiszar1978}. 
Holenstein and Renner \cite[Theorem 3]{Holenstein2006} proved that same results hold when the supremum is taken over all distributions $P_{UV}$ satisfying $X-U-V$. Holenstein in \cite[Theorem 3.3]{holenstein2006strengthening} showed that supremum
can always be achieved by a OM-SKA protocol and is finite. Therefore, we have the following theorem for the one-way secret key capacity.

%	But, an important lower bound on $C_s(X,Y|Z)$ is the forward key capacity given in \cite[Theorem 1]{Ahlswede1993}. This capacity is  \emph{one-message secret key} (OM-SK) Capacity $C_s^{om}(X,Y|Z)$, that is the maximum achievable SK rate, if we restrict parties to use only one public message for reconciliation. We also use $S_{\eps,\sigma}^{om}(X,Y|Z)$ to denote the largest achievable SK length when parties are restricted to use  one public message only. 
}

	\begin{theorem} [Theorem 1 of \cite{Holenstein2006}]
	    For any given source model $( X, Y, Z)$ with IID distribution 
		$P_{X^n Y^n Z^n}$, the OW-SK capacity is given by
		\begin{equation}\label{eq:OW-SK-Cap}
		    C_s^{ow}(X,Y|Z) = \max_{P_{UV|X}} H(U|Z,V) - H(U|Y,V),
		\end{equation}
		where optimization is over joint distributions  $P_{UV|X}$'s such that $X-U-V$ holds.
	\end{theorem}
	Finding 
	%For some source models, finding an
	 explicit solution to \eqref{eq:OW-SK-Cap} in general is not known. 
	 %might be difficult, while f
	 For some source models, the optimizing $P_{UV|X}$ can be analytically calculated \cite{Holenstein2006,holenstein2006strengthening}, and 
	 be used to construct OM-SKA 
	 % One-message SKAs can 
	 that achieves the OW-SK capacity \cite{Holenstein2006}.  
	
	\begin{corollary}[Corollary 4 of \cite{Renes2013}]\label{cor:lessnoisy}
	    For a source model with IID distribution $P_{X^n Y^n Z^n}$ such that for any RV $U$ satisfying $U-X-(Y,Z)$, we have $I(U;Y)\geq I(U;Z)$,
	    \begin{equation}
	        C_s^{ow}(X,Y|Z) = H(X|Z) - H(X|Y).
	    \end{equation}
	\end{corollary}
	
	A special case of Corollary~\ref{cor:lessnoisy} is when the Markov chain  $X-Y-Z$ holds. In this case, the OW-SK capacity is equal to  $C_s^{ow}(X,Y|Z) = H(X|Z) - H(X|Y) = I(X;Y|Z)$. 
	%That is for this case the OM-SK capacity is a tight lower bound. 

%   When the source  satisfies % For the special case where
%   the  Markov relation 
%   $X-Y-Z$, % holds 
%   and 
%   Let $(X^n,Y^n,Z^n)$  be   $n$ copies of an   IID source that satifies   the  Markov relation 
%   $X-Y-Z$. % holds 
% %   and 
% For this case the %it is proved that  
% secret key capacity 
%  is given by  %proved to be 
%  $C_s(X,Y|Z) = I(X;Y|Z)$. 
%where $I(X;Y|Z)$ is the conditional mutual information.  
%For  a general %case of a
% source  %model 
 %with $n-$IID copies $(X^n,Y^n,Z^n)$ however,
%The secret key capacity is not known \cite{ElGamal2011} when the source  does not satisfy Markov relation. %model 
 %with $n-$IID copies $(X^n,Y^n,Z^n)$
%The best known upper and lower-bounds on the SK capacity is given in %are due to Gohari and Anantharam in
% \cite{Gohari2010}.

\section{$\mathbf{\Pi}_{\mathrm{SKA}}$: A One-message SKA Protocol}\label{sec:ska}

Consider the source model setting, %of  Section \ref{sec:Source_Model}, 
    and assume Alice, Bob and Eve have 
    their corresponding    $n$ components of the source $(X^n,Y^n, Z^n)$. 
     Let    the  required secrecy and reliability parameters of the key be $\sigma$ and $\eps$, respectively.
           Alice and Bob choose  
   two 2-UHFs % universal hash
    $h_s:\mc X^n\to \{0,1\}^t$ and $\hat{h}_{s'}:\mc X^n\to \{0,1\}^\ell$,
    and share (over public channel)  two  uniformly random seeds $s\in\mc S$ and $s'\in\mc S'$ for the two families.
  (In the following we show how the values of 
    $t$ and $\ell$  will be  determined.)
    % by $n$, the source distribution $P_{XYZ}$,  $\sigma$ and $\eps$. %and reliability and secrecy parameters.
    %{\color{magenta}   Alice and Bob share (over public channel) uniformly random seeds $s\in\mc S$ and $s'\in\mc S'$ for the two families.}
    In the rest of this paper we use the  multiplicative hash family (See Section~\ref{sec: pre}) for the 2-UHF.
      %A concrete hash function family for binary source and binary output  uses multiplication over finite fields.
      
   %At a high level, o
   Our SKA protocol ($\mathbf{\Pi}_{\mathrm{SKA}}$) works as follows:
	Alice uses $h_s(\cdot)$  %he first hash function 
	to compute the hash value of 
	her sample vector $x^n$, and sends it to Bob;  Bob uses the received hash value, his sample vector $y^n$,
	 and the known probability distribution of the source,  
	to recover Alice's sample  vector {\em reconciliation}).    
The main idea behind the reconciliation technique of Protocol 1, used by Bob, is  to  divide  %based on 
%cutting
 the ??range of values (spectrum) ?? of $P_{X^n|Y^n=y^n}$ into two parts, and search in  only the main part to find Alice's vector.
%Bob then searches for Alice's sample in only the main slice. 
This  reduces search complexity at the cost of increased error.
%improves the search complexity by not searching the whole support of $P_{X^n|Y^n=y^n}$. 
By choosing  an appropriate value for $t$, Bob can bound the reconciliation error to $\eps$.
The transmitted hash value will also be used in conjunction with their vector $z^n$ to learn about the key, and so longer key hash values (reduced error probability) will result in shorter keys.
	 Alice and Bob will estimate the total  leaked information about their common strings, and
	 remove  it %this leaked information  
	 during the \textit{privacy amplification}  (key extraction) phase  that is implemented by using  a
	  second 2-UHF hash function $\hat{h}_{s'}.$
	 % in a \textit{privacy amplification}  (key extraction) phase. 

    The security and reliability of the protocol, and     the choice of   the hash functions parameters  %a choices
    that result in an $( \eps, \sigma)$- SKA %(reliability and secrecy of the  key at the desired levels,  are obtained using the proof of 
    are given in Theorem \ref{main}.  We prove the theorem for the case that the distribution $P_{X^n Y^n Z^n} = \Pi_j P_{X_j Y_j Z_j}$ is due to independent experiments which are not necessarily identical.  Corollary~\ref{cor:iid} shows that the resulting bound from the theorem can be tightened for IID case.
    % are proved respectively.

    \begin{algorithm}[t]
	    \DontPrintSemicolon
	    \caption{$\mathbf{\Pi}_{\mathrm{SKA}}$: A Capacity-achieving SKA %$(\mathbf{\Pi}_{\mathrm{SKA}})$
	    }
	    \label{alg: 1wayI}%\label{prot:SKA}
	    \SetKwInput{PorKnw}{\small{%System 
	    Public Information}}
	    \SetKwInput{PorAsm}{Assumption}
	    \SetKwInput{PorPar}{Parameter}
	    \PorKnw{$P_{XYZ}$}
	    \KwIn{%The 
	    $n-$fold samples 
	    $x^n\in\mc X^n$ and $y^n\in\mc Y^n$, $\epsilon$, $\sigma$.} 
	   % \PorPar{$\lambda$.} 
	    \KwOut{Key estimates $ k_A$ and  $ k_B$.}
	   % \PorPar{$\lambda$} 
	    \SetKwFor{MyFor}{for}{}{end~for}
	    \SetKwIF{If}{ElseIf}{Else}{if}{then}{else if}{else}{end~if}
	
	    \BlankLine
	
	    \tcp*[h]{%System 
	    Initiation Phase}	\BlankLine
        \textbullet~ Alice and Bob, (i) %\textbf{do}
         find and share $\lambda$,  and $\ell $ and $t$ for the hash functions $h_s: \mc X^n\to \{0,1\}^t$ and $\hat h_{s'}: \mc X^n\to \{0,1\}^\ell$, (ii) generate and share the
       % exchange
        seeds $s\in\mc S$ and $s'\in \mc{S}'$ for the hash function.
        % and parameters $l$ and $t$ to agree upon hash functions $h_s: \mc X^n\to \{0,1\}^t$ and $h_s': \mc X^n\to \{0,1\}^\ell$. 
	
	  %  \textbullet~  Alice and Bob \textbf{do} agree upon the parameter $\lambda$. 
	
	    \BlankLine
	    \tcp*[h]{Information Reconciliation Phase}	\BlankLine
	
	    \nl
        Alice sends the hash value $v=h_s(x^n)$ to Bob.
        
	    \nl
	    Bob forms a list of guesses for $x^n$, %i.e.,
	   \begin{equation} \label{vectorbound}
	   \mathcal{T}(X^n|y^n) = \{\hat{x^n} :-\log P_{X^n|y^n} (\hat{x^n}|y^n) \leq \lambda\}.
	   \end{equation}

	    \nl
        Bob finds, if exists, a unique $\hat{x^n}\in \mathcal{T}(X^n|y^n)$ such that $h_s(\hat{x^n}) = v$.
    	
	    \nl
        \If{no $\hat{x^n}$ was found \textbf{or} $\hat{x^n}$ was not unique}{Abort the protocol.}
    
        \BlankLine
	    \tcp*[h]{Key Extraction Phase}	\BlankLine
		
	    \nl
	    Alice computes  $ k_A=\hat{h}_{s'}(x^n)$. 
		
	    \nl
	    Bob computes $ k_B=\hat{h}_{s'}(\hat{x^n})$.
	
    \end{algorithm}

\begin{theorem}\label{main}
Let the source model $(X^n,Y^n,Z^n)$  be described by a joint distribution $P_{X^n Y^n Z^n} = \Pi_j P_{X_j Y_j Z_j}$.  
% Let %Given that 
% Alice, Bob and Eve have their corresponding   % access to %a source model of
%  $n$ samples of  an IID source $(X^n,Y^n,Z^n)$  with %the known
%   joint distribution $P_{XYZ}$, that is assumed public, and 
%  satisfying the Markov relation $X-Y-Z$. %
For  any pair of real numbers $(\eps, \sigma)$ where $0<\eps,\sigma<1$,  %and with the choice of  $\lambda=H(X^n|Y^n)+n\delta,$ and $\eps/2=2^{\frac{-n\delta^2}{2\log^2(|\mathcal{X}|+3)}}$, 
  $\Pi_{\mathrm{SKA}}$ %the proposed SKA %non-interactive
Protocol~\ref{alg: 1wayI} gives  % can achieve 
an $(\eps,\sigma)$-SK with maximum key length of, % the maximum length of 
\begin{align}\nonumber
            &\max(\ell_{\eps,\sigma}^{\mathbf{\Pi}_{\mathrm{SKA}}}) = H(X^n|Z^n)-H(X^n|Y^n)\\ 
            &\qquad\qquad    -\sqrt{n}f_{\eps,\sigma}(|\mc X|) 
             - \log \frac{4n^3}{\eps\sigma^2} + \mc O(\frac{1}{\sqrt{n}}),
        \end{align}
	where $f_{\eps,\sigma}(|\mc X|) =\sqrt{2}\log(|\mc X|~+~3)\left( \sqrt{\log 1/\eps} + \sqrt{\log 2/\sigma} \right)$.
\end{theorem}

%%%%%%%%%%%%%%%%%

\begin{proof}
%\underline{First} 

{\em Reliability.} 
Bob's recovery algorithm  searches the set $\mathcal{T}(X^n|y^n)$ for vector(s) that their hash value match the received hash value, 
and fails in two cases: (i)    $x$ is not in the set, and (ii) there are more than one vector in the set whose hash value matches the received hash value $v$.
Bob's  failure probability $P_e=\pr{K_A\neq K_B}$ is upper bounded by finding the probabilities of the above two events.
% of the protocol is due to two types of bad events.
Each event corresponds to the
%set of 
possible %  There are two possible sets of 
samples of Alice, as shown below. % by Alice that cause failure of the protocol:
\begin{eqnarray*}
&&\xi_1= %\{x^n\notin \mc{T}(X^n|y^n)\}=
\{x^n :-\log P_{X^n|Y^n} (x^n|y^n) > \lambda\} \\
&&\xi_2=\{x^n\in \mc{T}(X^n|y^n) :\exists\ \hat{x^n}\in\mc{T}(X^n|y^n)\ \mathrm{s.t.}~\\
&&\qquad\qquad\qquad\qquad\qquad h_S(\hat{x^n})=h_S(x^n)\}.
\end{eqnarray*}

%For the given choice of $\lambda$ and $\delta$ (

To bound $\pr{\xi_1}$,   we use the result of \cite[Theorem 2]{Holenstein2011}, for $n$-IID samples of a joint distribution that states:
\[ \ \Pr[- \log P_{X^n|Y^n} (x^n|y^n) > H(X^n|Y^n)+n\delta\} ]\leq \beta. \] 
 for $\beta=2^{-\frac{n\delta^2}{2\log^2(|X|+3)}}$. 
By choosing 
\begin{equation}\label{eq:lambda}
    \lambda=H(X^n|Y^n)+n\delta_1,
\end{equation}
where %and
$\delta_1$ satisfies, ${\eps_1}=2^{\frac{-n(\delta_1)^2}{2\log^2(|\mathcal{X}|+3))}}$  for some chosen value of $\eps_1\leq \eps$,
we   have 
$\xi_1=\{x^n :-\log P_{X^n|Y^n} (x^n|y^n) > H(X^n|Y^n)+n\delta_1\},$
%Therefore, from \cite[Theorem 2]{Holenstein2011}, 
and $\pr{\xi_1}\leq {\eps_1}$.
%Now 

To bound $\pr{\xi_2}$,  we note that for   a ${\hat{x^n}}\in \mc{T}(X^n|y^n)$, 
 the collision probability with any other ${x^n} \in {\cal X}^n$
is bounded by
$\pr{h_S (\hat{x^n} )= h_S(x^n)} \leq 2^{-t}$ (Definition \ref{def: UHF}),  and so the total probability 
that some element  $\mc{T}(X^n|y^n)$  collides with an element in ${\cal X}^n$  is  $|\mathcal{T}(X^n|y^n)|\cdot2^{-t}$.
 That is $$\pr {\xi_2} \leq |\mathcal{T}(X^n|y^n)|\cdot2^{-t}.$$
On the other hand, since the probability of each element of $ \mathcal{T}$ is bounded by   $2^{-\lambda}$, we have 
 $|\mathcal{T}(X^n|y^n)|. 2^{-\lambda}\leq \pr{\mathcal{T}(X^n|y^n))}\leq 1$, and we have $|\mathcal{T}(X^n|y^n)|\leq 2^{\lambda}$.  
Let $t=\lambda-\log{\eps_2}$  for some chosen value of $\eps_2\leq \eps$. Then we have
$\pr{\xi_2}\leq {\eps}_{2}$. That is
\begin{equation}\label{eq:t}
    t=H(X^n|Y^n)+n\delta_1-\log\eps_2
\end{equation}

The above shows that the  choice  of $\eps_1$ determines $\delta_1$, and then $\lambda$, which together with the choice of $\eps_2$, determines $t$.
Finally,  $\eps = \eps_1+\eps_2$  and  $P_e\leq \eps$.
%For $\eps_1+\eps_2=\eps$ by the above argument $P_e$ is bounded with $\eps$ i.e. $P_e\leq \eps$. 
%From 
Equation (\ref{eq:t}) clearly shows the relation between $t$ and the error probability:  smaller $\eps_1$ and $\eps_2$ %(which in turn is related to $\eps_1$)
give % we see that less error probability corresponds
 larger $t$. This is expected as larger $t$ %what we expect as the larger $t$ 
provides  more information about $X^n$ to Bob for reconciliation.   We also note that larger information about $X^n$ reduces achievable key length of the protocol.

%$\eps$ goes to 0 as $n$ grows and $2^{-\gamma}$ can become arbitrary small by choosing large enough $\gamma$. 
%This completes the reliability argument.

\noindent
{\em Key secrecy.}
%
%\underline{Second,}
To show the secrecy of the key, we need to bound the statistical distance of the joint distribution of the derived key and the adversary's information,  from %with the 
the joint distribution of the uniform distribution and the adversary's information. 
%by $\sigma$ to prove the security of the key. 

According to the Lemma~\ref{smooth-hash} in the appendix, we have 
\begin{align}\nonumber
  &\sd((\hat{h}_{S'}({X^n}),h_S({X^n}),S',S,{Z^n});\\\label{eq:sigma-sec-smooth}
  &(U_\ell,h_S({X^n}),S',S,{ Z^n}))\leq 2\eps'+\frac{1}{2}\sqrt{2^{t+\ell-\tilde{H}^{\eps'}_\infty(X^n|Z^n)}},
\end{align}
We now note the following:

(ii) The relation between the smooth average min-entropy and the smooth conditional min-entropy is  given in \cite[Appendix B]{Dodis2008}, and  states $\tilde{H}^{\eps'}_\infty(X^n|Z^n)\geq H_{\infty}^{\eps'}(X^n|Z^n)$.
%(see Lemma~\ref{smooth} in the appendix WHAT DOES THE LEMMA SAY , and

(iii) Using \cite[Theorem 1]{Holenstein2011},  we have $H_{\infty}^{\eps'}(X^n|Z^n)\geq H(X^n|Z^n)-n\delta'$, with $\eps'=2^{\frac{-n\delta'^2}{2\log^2(|\mathcal{X}|+3)}}$. Therefore, we can substitute $\tilde{H}^{\eps'}_\infty(X^n|Z^n)$ in (\ref{eq:sigma-sec-smooth}) with $H(X^n|Z^n)-n\delta'$:
\begin{align*}
  \sd((&\hat{h}_{S'}({X^n}),h_S({X^n}),S',S,{Z^n});\\&(U_\ell,h_S({X^n}),S',S,{Z^n}))\leq 2{\eps'}+\frac{1}{2}\sqrt{2^{t+\ell-{H}(X^n|Z^n)+n\delta'}},  
\end{align*}
where $\eps'=2^{\frac{-n\delta'^2}{2\log^2(|\mathcal{X}|+3)}}$ and thus 
$
 \delta'=\sqrt{\frac{(\log\frac{1}{\eps'})(2\log^2(|\mathcal{X}|+3))}{n}}.
$
% \begin{align*}
% \eps'=2^{\frac{-n\delta'^2}{2\log^2(|\mathcal{X}|+3)}} &\Rightarrow (\log \eps')(2\log^2(|\mathcal{X}|+3))=-n\delta'^2 \\
%  &\Rightarrow \delta'=\sqrt{\frac{(\log\frac{1}{\eps'})(2\log^2(|\mathcal{X}|+3))}{n}}.
% \end{align*}
To satisfy $\sigma$-secrecy of the key agreement protocol it is sufficient to have 
$2{\eps'} +\frac{1}{2}\sqrt{2^{t+\ell-{H}(X^n|Z^n)+n\delta'}} \leq \sigma.$
Thus, the % requirement gives t
 maximum achievable key length of the %with the proposed 
protocol must satisfy,
\begin{align*}
&\sqrt{2^{t+\ell-{H}(X^n|Z^n)+n\delta'}}\leq 2(\sigma-2{\eps'})\\
%&\quad \Rightarrow t+\ell-{H}(X^n|Z^n)+n\delta' \leq 2+2\log(\sigma-\eps')\\
&\quad \Rightarrow \ell \leq H(X^n|Z^n)-n\delta'-t+2+2\log(\sigma-2\eps')
\end{align*}
%For $t=\lambda+\log\frac{2}{\eps}$ we have
%\begin{align*}
%\ell \leq H(X^n|Z^n)-&H(X^n|Y^n)-2n\delta\\
%&+2+2\log(\sigma-\frac{\eps}{2})-\log\frac{2}{\eps}
%\end{align*}

To have error probability bounded by $\epsilon$, %satisfy reliability requirement  as well 
we will use $t$ from (\ref{eq:t}) and obtain the 
maximum achievable key length with $(\eps, \sigma)$ parameters, as % is then
\begin{align}\label{bound1}
\ell \leq& H(X^n|Z^n)-H(X^n|Y^n)+2+\log \eps_2(\sigma-2{\eps'})^2 \nonumber\\
&-\sqrt{2n} \log(|\mathcal{X}|+3)(\sqrt{(\log \frac{1}{\eps'})}+\sqrt{(\log \frac{1}{\eps_1})}).
%&\qquad +2+2\log(\sigma-\eps)-\log\frac{2}{\eps}.
\end{align}
We will then optimize the bound by choosing appropriate values for $\eps'$,  $\eps_1$ and  $\eps_2$.

$\eps'$ is the smoothing parameter and can  be  chosen arbitrarily  subject to satisfying $\eps'\leq\frac{\sigma}{2}$. 
The effect of $\eps'$ on the upper bound  (\ref{bound1}) appears  in two terms:
larger values of $\eps'$,  in $\log(\sigma-2\eps')$ reduce the RHS of the bound, and  
 as part of the coefficient of the term $-\sqrt{n}$, increase the RHS of the bound.
 However the effect of the latter will be multiplied by %is higher because
  the $\sqrt{n}$, and so we will choose 
%
%From (\ref{maxkey})?? WRONG REF?  we see that larger $\eps'$ corresponds to  smaller coefficient for ($-\sqrt{n}$) that increase the length of the key and on the other hand,  corresponds to smaller $\log(\sigma-2\eps')$ that decreases the key length.  However, we let $\eps'$ to be small as the effect of the coefficient of $\sqrt{n}$ on the length of the key is more. That is we 
 $\eps'=\frac{n-1}{2n}\sigma$.
 
For $\eps_1$ and  $\eps_2$,  we have $\eps_1+\eps_2\leq \eps$.  In the bound  (\ref{bound1}) $(\log \frac{1}{\eps_1})$ is the coefficient of $\sqrt{n}$, and  its larger values correspond  to larger value of  the RHS of the bound. $\eps_2$ however appears within a constant term.
We thus choose $\eps_1=\frac{n-1}{n}\eps$ and $\eps_2=\frac{\eps}{n}$. 
These are reasonable choices for bounding error probabilities of $\xi_1$ and $\xi_2$:
%  is reasonable with regards to the initial definition of $\xi_1$ and $\xi_2$ because 
$\xi_1$ is the error of  $x^n$ being outside of $\mc T(X^n|y^n)$ while $\xi_2$ is related to the collision probability of the hash function,
 and is expected to be much smaller than $\xi_1$.
\remove{and and the choice of $\eps_1$
For reliability parameters we see larger $\eps_1$ and $\eps_2$ corresponds to larger key length but since $\eps_1+\eps_2\leq \eps$, the increase of one decreases the other one. We let $\eps_1$ to be much larger than $\eps_2$ because $(\log \frac{1}{\eps_1})$ is the coefficient of $\sqrt{n}$, while $\log \eps_2$ is a constant term. Therefore, we set $\eps_1=\frac{n-1}{n}\eps$ and $\eps_2=\frac{\eps}{n}$. This choice is reasonable with regards to the initial definition of $\xi_1$ and $\xi_2$ because $\xi_1$ is the error that happens when $x^n$ is outside of $\mc T(X^n|y^n)$ while $\xi_2$ corresponds to the collision probability of the hash function and is supposed to be much smaller than $\xi_1$.
}
Using these substitutions we have
\begin{align} \nonumber
&\max(\ell_{\eps,\sigma}^{\mathbf{\Pi}_{\mathrm{SKA}}}) = H(X^n|Z^n)-H(X^n|Y^n)
 + 2+\log \frac{\eps \sigma^2}{n^3}\\\label{maxkey} 
 &-\sqrt{2n} \log(|\mathcal{X}|+3) \big(\sqrt{\log {\frac{n}{(n-1)\eps}}}+\sqrt{ \log {\frac{2n}{(n-1)\sigma}}}\big). 
%&\qquad +2+2\log(\sigma-\eps)-\log\frac{2}{\eps}.
\end{align}

Since $\sqrt{\log\frac{an}{(n-1)b}} = \sqrt{\log\frac{a}{b}} + \mc O(1/n)$, we have
\begin{align}\nonumber
            &\max(\ell_{\eps,\sigma}^{\mathbf{\Pi}_{\mathrm{SKA}}}) = H(X^n|Z^n)-H(X^n|Y^n) \\ \label{eq:OM-length-2}  &\qquad\qquad -\sqrt{n}f_{\eps,\sigma}(|\mc X|) 
             - \log \frac{4n^3}{\eps\sigma^2} + \mc O(\frac{1}{\sqrt{n}}),
        \end{align}
which completes the proof.
\end{proof}

\begin{remark}
 The third order term of these bounds can be further improved by choosing $\eps_1,\eps_2,$ and $\eps'$ differently. For example, let $\eps'=\frac{\sqrt[4]{n} - 1}{2\sqrt[4]{n}}\sigma$, $\eps_1 = \frac{\sqrt{n} - 1}{2\sqrt{n}}\eps$, and $\eps_2 = \frac{\eps}{\sqrt{n}}$. Then 
    \begin{align}\nonumber
            &\max(\ell_{\eps,\sigma}^{\mathbf{\Pi}_{\mathrm{SKA}}}) = H(X^n|Z^n)-H(X^n|Y^n)\\  \label{eq:OM-main}
            &\qquad\qquad    -\sqrt{n}f_{\eps,\sigma}(|\mc X|) 
             - \log n + \mc O(1),
    \end{align}
    which follows from the fact that $$\sqrt{\log\frac{a\sqrt{n}}{(\sqrt{n}-1)b}} = \sqrt{\log\frac{a}{b}} + \frac{1}{2\ln{2}\sqrt{n\log\frac{a}{b}}} + \mc O(\frac{1}{n}).$$
\end{remark}

Theorem~\ref{main} gives 
the maximum achievable key length $\ell_{\eps,\sigma}^{\mathbf{\Pi}_{\mathrm{SKA}}}$ %$\max(\ell_{\eps,\sigma}^{\mathbf{\Pi}_{\mathrm{SKA}}})$, 
of the protocol and %hich is an achievable one-message key length and a
provides a lower bound on the maximum key length of OW-SKA protocols.
%highest {OM-SK} length, i.e.,
That is,
 $S_{\eps,\sigma}^{ow}(X^n,Y^n|Z^n)\geq \max(\ell_{\eps,\sigma}^{\mathbf{\Pi}_{\mathrm{SKA}}})$. 
Next corollary tightens the lower bound %,  makes  this lower-bound tighter (more accurate) for 
for  IID sources, % distributions, by the 
using  Berry-Essen inequality \cite{berry1941accuracy,esseen1945fourier}. % \cite{full-v1}.
    \begin{corollary}\label{cor:iid}
        For any source model described by IID distribution $P^n=\Pi_j P_{X_j Y_j Z_j}$ we have 
        \begin{align}\nonumber
            S_{\eps,\sigma}^{ow}(X^n,Y^n|Z^n) &\geq 
            n(H(X|Z) - H(X|Y)) \\ \label{eq:OM-length-1}  &\qquad -\sqrt{n}g_{\eps,\sigma} - \frac{3}{2}\log n + \mc O(1),
        \end{align}
        where $$g_{\eps,\sigma} = Q^{-1}(\eps)\sqrt{\Delta_{X|Y}} + Q^{-1}(\frac{\sigma}{2})\sqrt{\Delta_{X|Z}},$$
    and $\Delta_{U|V} = \myvar{-\log P_{U|V}}$. 
    
    \end{corollary}
        \begin{proof}
            We revisit the proof of Theorem~\ref{main}. For reliability we bound the probability of these two events: 
    \begin{eqnarray*}
    &&\xi_1= %\{x^n\notin \mc{T}(X^n|y^n)\}=
    \{x^n :-\log P_{X^n|Y^n} (x^n|y^n) > \lambda\} \\
    &&\xi_2=\{x^n\in \mc{T}(X^n|y^n) :\exists\ \hat{x^n}\in\mc{T}(X^n|y^n)\ \mathrm{s.t.}~\\
    &&\qquad\qquad\qquad\qquad\qquad
    h_S(\hat{x^n})=h_S(x^n)\}.
    \end{eqnarray*}
    
    Let $W_i = -\log P_{X_i|Y_i}$ and $\lambda = nH(X|Y) + \sqrt{nV_{X|Y}}Q^{-1}(\eps-\theta_n)$, where $\Delta_{X|Y} = \myvar{-\log P_{X|Y}}$, and $\theta_n = \frac{1}{\sqrt{n}} + \frac{3\rho}{V^{3/2}_{X|Y} \sqrt{n}}$. Then by Lemma~\ref{lem:Berry-Esseen},  $\pr{\xi_1} \leq \eps - \frac{1}{\sqrt{n}}$. By choosing $t=\lambda -\log \frac{1}{\sqrt{n}}$, we will get $\pr{K_A\neq K_B} \leq \pr{\xi_1}+\pr{\xi_2} \leq \eps$.

    For the secrecy constraint, we use Lemma~\ref{smooth-hash}. By this leftover hash lemma and noting the fact that $\tilde{H}^{\eps'}_\infty(X^n|Z^n)\geq H_{\infty}^{\eps'}(X^n|Z^n)$, we have $\sqrt{2^{t+\ell - H^{\eps'}_{\infty}(X^n|Z^n) }} \leq 2(\sigma - 2\eps')$, for any $\eps'$. This implies that for $\eta_n = \frac{2}{\sqrt{n}}$ we get
    \begin{align*}
        \ell \leq H^{\frac{\sigma - \eta_n}{2}}_{\infty}(X^n|Z^n) -  t + \log 4\eta_n^2.
    \end{align*}
    From \cite{Hayashi2019} we know that for IID distribution $P_{X^n Z^n}$, $$H^{\delta}_{\infty}(X^n|Z^n) = n H(X|Z) -Q^{-1}(\delta)\sqrt{n\Delta_{X|Z}} + \mc O(1), $$
    where $\Delta_{X|Z} = \myvar{-\log P_{X|Z}}$. Thus, 
    \begin{align*}
        \ell \leq & n(H(X|Z)- H(X|Y) )  \\
        & \sqrt{n}\left( Q^{-1}(\eps - \theta_n)\sqrt{\Delta_{X|Y}} + Q^{-1}(\frac{\sigma - \eta_n}{2})\sqrt{\Delta_{X|Z}} 
        \right) \\
        &- \frac{3}{2}\log n + \mc O(1).
    \end{align*}
    And thus the proof is complete by using Taylor expansions to remove $\theta_n$ and $\eta_n$. 
    \end{proof}   
    \begin{corollary}[%On achieving the 
    OW-SK Capacity]
      For  IID  distribution $P_{X^n Y^n Z^n}$, %is IID and %such that for any RV $U$ satisfying $U-X-(Y,Z)$, we have $I(U;Y)\geq I(U;Z)$,
        Protocol $\mathbf{\Pi}_{\mathrm{SKA}}$ achieves the OW-SK capacity.   
    \end{corollary}
    %

    % \noindent See the proof in \cite{full-v1}. 
    \begin{proof}
    If variables $U$ and $V$ can be found for a given source model $(X,Y,Z)$ with distribution $P_{XYZ}$, such that $X-U-V$ holds and $P_{UV|X}$ maximizes $H(U|Z,V)-H(U|Y,V)$, then the protocol  $\mathbf{\Pi}_{\mathrm{SKA}}$ achieves $C_s^{ow}(X,Y,Z)$, i.e., the OW-SK capacity. To prove this, first note that $\mathbf{\Pi}_{\mathrm{SKA}}$ achieves the SK rate of $R=H(X|Z) - H(X|Y)$ for any given source model $(X,Y,Z)$. Due to this protocol, parties first reconcile on $X^n$ and then extract the key from $X^n$ knowing the fact that the adversary has access to a correlated variable $Z^n$. Assume for a given source model $(X,Y,Z)$ the optimal variables $U$ and $V$ can be calculated such that $X-U-V$ holds, and $C_s^{ow}(X,Y,Z) = H(U|Z,V)-H(U|Y,V)$. In the IID regime, parties observe IID variables $(X^n,Y^n,Z^n)$. Thus, in the first step, Alice, who has access to $X^n$, generates $U^n$ and $V^n$. Then she broadcasts $V^n$ over the public channel. Note that now Eve also has access to both side information variables $Z^n$ and $V^n$. After these initial steps, Alice and Bob run Protocol $\mathbf{\Pi}_{\mathrm{SKA}}$ to reconcile on $U^n$. For the reconciliation step Bob uses his side information, i.e., $(Y^n,V^n)$ to find $U^n$. To extract the key from $U^n$, parties know that Eve has access to $(Z^n,V^n)$, thus by performing the appropriate key extraction the final SK rate will be  $H(U|Z,V)-H(U|Y,V)$, which is equal to $C_s^{ow}(X,Y,Z)$. Hence, the proof is complete. This proof is due to \cite{Ahlswede1993}, Section IV. Also see \cite{Renes2013}, Section III.B. 
    \end{proof}
    
    In Figure~\ref{fig:comp} %he following example
     we compare the two lower bounds of \eqref{eq:OM-length-1} and the  lower bound given in \eqref{eq:OM-main} for
 a source model where $X$ is a uniformly distributed binary  variable, $Y=\mathrm{BSC}_p(X)$, 
     and $Z=\mathrm{BSC}_q(Y)$ are obtained as the output of binary symmetric channels on $X$ and $Y$, respectively,  and $\mathrm{BSC}_p(\cdot)$ denotes a binary symmetric channel with crossover probability $p$.
 %Here, $\mathrm{BSC}_p(\cdot)$ denotes a binary symmetric channel with crossover probability of $p$. 
 For this example, the OW-SK capacity is  $C_{s}^{ow} = h_2(p*q) - h_2(p)$, where $p*q = p(1-q) + (1-p)q$, and $h_2(\cdot)$ is the binary entropy given by $h_2 (p) = -p\log p - (1-p)\log(1-p)$.  For $p=0.02,$ $q=0.15$, and $\eps=\sigma=0.05$, we have  $C_{s}^{ow}= 0.5$. The maximum secret key length by $\mathbf{\Pi}_{\mathrm{SKA}}$, achieves this OW-SK capacity. 
    The finite blocklength bounds % approximations 
    of \eqref{eq:OM-length-1} and \eqref{eq:OM-main}  are converted to key rate (divided by $n$), and %calculated
are depicted for this example  in Figure~\ref{fig:comp}.  
   %  The bounds are divided by $n$ to show the gap to the one-way SK capacity. As you can see, 
   The graph shows the bound of \eqref{eq:OM-length-1} is tighter than the bound of \eqref{eq:OM-main}, and is  closer to the capacity. % upper bound t 
   %As we will discuss later i
   In Section~\ref{sec:comparison},  % both of these bounds are far tighter than
   we derive  the bounds associated with the constructions in %given in 
   \cite{holenstein2006strengthening}.  These bounds will not have positive values for the range of $n$ used in Figure~\ref{fig:comp} and so are not included.

\begin{figure}
	\centering
	\includegraphics[width=0.95\linewidth]{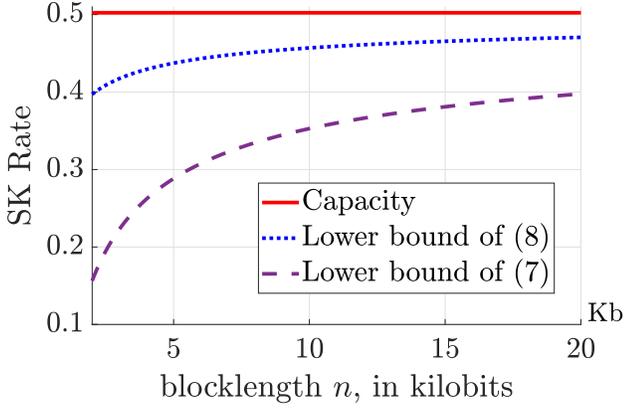}
	\caption{Comparing the finite-length bounds of \eqref{eq:OM-length-1} and \eqref{eq:OM-main}. Here, $P_{X^n Y^n Z^n}$ is IID, $X$ is binary uniform, $Y=\textrm{BSC}_{0.02}(X)$, $Z=\textrm{BSC}_{0.15}(Y)$, reliability and secrecy parameters are $\eps=\sigma=0.05$, and the secret key capacity is $0.5$. The sample's length $n\in[2000,20000]$. For this range of $n$, the bounds of \eqref{eq:LB-HR-1} and  \eqref{eq:LB-HR-2} give SK rate of $0$. In this example, bounds of \eqref{eq:LB-HR-1} and  \eqref{eq:LB-HR-2} are  positive only for $n>10^{7}$, and $n>1.1\times 10^{9}$, respectively.  }\label{fig:comp}
\end{figure}

    % The corollary follows from expression  (\ref{maxkey}), by dividing both sides by $n$ and %Theorem~\ref{main} , when 
    % $n\rightarrow \infty$.
    
    \begin{corollary}[Public communication cost]\label{cor:pub-mess}
        For the  source model described by distribution $P_{X^n Y^n Z^n} = \Pi_j P_{X_j Y_j Z_j}$, let $F^{om}_{\eps,\sigma}$ denote the
        {\em least} % optimum 
        communication cost (in bits) that is needed for a OM-SKA  $S_{\eps,\sigma}^{ow}(X^n, Y^n|Z^n)$.
 \remove{
     public communication protocol that uses the least number of public communication bits $L(F^{ow}_{\eps,\sigma}) =\log|F^{ow}_{\eps,\sigma}|$ for reconciliation which is required for achieving $S_{\eps,\sigma}^{ow}(X^n, Y^n|Z^n)$. 
     }
     Then %for %for any $\eps_1+\eps_2 = \eps$, 
        % $L(\vect F) =\log|\vect F|$ %is upper bounded by
        \begin{align} \label{eq:FL-F}
         %    L(F^{ow}_{\eps,\sigma})
      F^{ow}_{\eps,\sigma}     \leq H(X^n|Y^n) + \sqrt{n}B_1^\eps(|\mc X|) 
            + \frac{1}{2}\log n + \mc O(1),
        \end{align}
        where $B_1^\eps(|\mc X|) = \sqrt{2}\log(|\mc X|+3)\sqrt{\log\frac{1}{\eps}}$.
        Moreover, for the case of IID source distributions, we have 
        \begin{align} \label{eq:FL-F-2}
          %   L(F^{ow}_{\eps,\sigma}) 
          F^{ow}_{\eps,\sigma} \leq
             n H(X|Y) + \sqrt{n}B_2^\eps 
            + \frac{1}{2}\log n + \mc O(1),
        \end{align}
        where $B_2^\eps =  Q^{-1}(\eps)\sqrt{\Delta_{X|Y}}$.
    
    \end{corollary}

\section{Comparison with related protocols}\label{sec:comparison}

We compare %the performance of 
Protocol \ref{alg: 1wayI} with  %some 
other known % proposed
 OM-SKA  protocols. 
 We compare the protocols based on the type of reconciliation, %The comparison criteria are:
  SK length, public communication cost, computation complexity for Alice and %computation complexity for 
  Bob, individually.  
  A summary of this comparison is given in Table~\ref{tab:comp}.  We  list these protocols in the following:

\vspace{1mm}
\noindent  \textbf{HR05~} Holenstein and Renner proposed a key agreement method in \cite{Holenstein2006} that achieves the OW-SK capacity of 
{\em a general distribution}. %any given distribution. 
The reconciliation message uses % phase of their protocol is based on
 linear codes.  We derive  two %Two finite-length
 finite blocklength lower bounds for the two % bounds can be derived based on two 
 variations of their SKA  \cite{holenstein2006strengthening}, one using a random linear code, and the second a concatenation of a linear code with a Reed-Solomon code.
 % We recalculated these bounds to present them as 
 The bounds 
% are shows as functions of $\eps$ and $\sigma$. The bounds are 
 given in Theorem 3.13, and Theorem 3.15 of \cite{holenstein2006strengthening}, are re-derived in the following proposition  as functions of $\eps$ and $\sigma$. In \cite{holenstein2006strengthening}, the bounds are expressed  in the form of $B(n) = n(C_s^{ow}-f_B(\kappa_1,\kappa_2))$, where $n\kappa_1 = \log(1/\eps)$ and $n\kappa_2 = \log(1/\sigma)$. 
 %The bounds are 
%pages 48, and 50, 
% by choosing $n\kappa_1 = \log(1/\eps)$ and $n\kappa_2 = \log(1/\sigma)$  ?? in their construction?? WHERE ARE THESE  VA:LUES

\begin{proposition}\label{prop:1}
    For any %given 
    source model with IID distribution $P_{XYZ}$, let $R_n=H(X^n|Z^n) - H(X^n|Y^n)$.  Then for large enough $n$ and any $\eps,\sigma<1/4$, we have
    \begin{align}\label{eq:LB-HR-1}
        S_{\eps,\sigma}^{ow}(X^n,Y^n|Z^n) &\geq [  R_n - \sqrt{n}f'_{\eps,\sigma} ]^+, \\
        \label{eq:LB-HR-2}
        S_{\eps,\sigma}^{ow}(X^n,Y^n|Z^n) &\geq  [  R_n - \sqrt[4]{n^3}
        g''_{\eps,\sigma} - \sqrt{n}f''_{\eps,\sigma} ]^+,
    \end{align}
    where $[a]^+ = \max\{ 0, a\}$, and 
    \begin{align*}
       f'_{\eps,\sigma} &= 90\log(|\mc X||\mc Y|)(\sqrt{\log 1/\eps} + \sqrt{\log 1/\sigma}), \\
       g''_{\eps,\sigma} &= \sqrt[4]{2^{22}\log(1/\eps)\log^2(|\mc X|)\log^2(|\mc X||\mc Y|)}, \\
       f''_{\eps,\sigma} &= 8\log(|\mc X|)\sqrt{\log(1/\sigma)}.
    \end{align*}
\end{proposition}

The first bound \eqref{eq:LB-HR-1} corresponds to random linear codes and second bound is due to concatenated codes.
%, see \cite{holenstein2006strengthening} for more details. 
For both lower bounds \eqref{eq:LB-HR-1} and \eqref{eq:LB-HR-2} the SKA protocol uses $\mc O(n^2)$ bits of communication. The computation complexity of Alice corresponding to both bounds is in $\mc O(n^2)$. The computation complexity of Bob  is in $\mc O (n^2)|\mc X|^n$ and $\mc O (n^2)$, respectively corresponding to  \eqref{eq:LB-HR-1} and \eqref{eq:LB-HR-2}. As it is mentioned in \cite{Holenstein2006,Renes2013}, the computation complexity of \eqref{eq:LB-HR-2} for Alice and Bob is \emph{efficient} (i.e., in $\mc O(n^d)$) but it is not \emph{practically efficient} (i.e., it is not in $\mc O(n)$ or $\mc O(n\log n)$). 

%Here, w
We note that our derived finite-length bounds in \eqref{eq:OM-length-1} and \eqref{eq:OM-main}, are far closer to the capacity upper bound than the finite-length bounds of HR05. For instance, considering the same setting and parameters of the example given for Fig~\ref{fig:comp}, the rate associated with %approximation of 
\eqref{eq:LB-HR-2} (i.e., bound of \eqref{eq:LB-HR-2} divided by $n$)  will be positive %gives positive rates only for
 $n>1.1\times 10^{9}$.

\vspace{1mm}
\noindent \textbf{RRS13~} Renes et. al proposed an SKA protocol that used polar codes for both reconciliation and privacy amplification \cite{Renes2013}.  The implementation cost of the protocol is $\mc O(n\log n)$ for both Alice and Bob, but the code construction for any given distribution might 
not be straighforward. %be difficult.
 The protocol uses a single message of length $\mc O(n)$. Their analysis of the protocol does not  provide %give any
  finite-length approximation of  the key length. 

\vspace{1mm}
\noindent \textbf{CBA15~} In \cite{Chou2015a},  authors propose an %another 
SKA protocol using % which utilizes
 polar codes. The reconciliation and privacy amplification is combined in one step polar coding, and the protocol requires a small pre-shared secret seed  of length $\mc O(2^{-a.n})$. This protocol %also 
 uses $\mc O(n)$
bit of public communication, and the analysis of \cite{Chou2015a} does not give any finite-length approximations for the final key length.

\remove{
We note that finding finite-length approximations for SKA with polar coding techniques, is an interesting open problem for future work, since the asymptotic computational complexity of these protocols are practically efficient. 

}
% \vspace{1mm}
% \noindent - \textbf{HTW16~}Hayashi et. al focused on the problem of secret key agreement in the degraded source model. They proposed an interactive protocol that achieves the SK capacity. HTW16 uses spectrum slicing and random binning techniques for the reconciliation phase. They give the second order approximation of the key length for this model and show that their protocol is optimum up to the second order term. The protocol uses many rounds of interaction $\mc O(n)$, and have computation complexity of $\mc O(2^n)$ for Bob and $\mc O(n^2)$ for Alice. The public communication bits required for this protocol is in $\mc O(n)$.    

\vspace{1mm}
\noindent \textbf{$\mathbf{\Pi}_{\mathrm{SKA}}$ (Protocol \ref{alg: 1wayI})~} This protocol uses universal hash functions for both reconciliation and privacy amplification. This protocol is very efficient in terms of public communication, as it uses a single message of length $\mc O(nH(X|Y))$ (See Corollary~\ref{cor:pub-mess}). 
This protocol gives achievability finite-length bounds as given in \eqref{eq:OM-main} and \eqref{eq:OM-length-1}. 
%Since, the implementation of hash functions is not hard in practice the protocol is easier for implantation, especially when $X$ and $Y$ are highly correlated. 
The computation cost of Alice is practically efficient; i.e., $\mc O(n\log\log n)$ (computing a single hash value). But, unfortunately for Bob, the computation cost is in $\mc O(2^{n.H(X|Y)})$; i.e., the implementation is not efficient.

% {\color{magenta}
% \textbf{-Wiretap Channel OM-SKA:} The wiretap message encoding protocol of \cite{sharifian2018post} uses universal hash function to encode a message and securely transmit it over the wiretap channel. A uniform key can be transmitted instead of the message for OM-SKA. The communication cost in this case is $\mathcal{O}(n)$ and the computational cost of Alice and Bob is the hash calculation cost which is $\mathcal{O}(n^2)$. The achievable key rate is given in \cite{sharifian2018post}. The protocol gaurantees security for symmetric channels.

% -\textbf{OM-SKA from Fuzzy Extractors.} By using fuzzy extractor of Fuller et al. \cite{fuller2016fuzzy} one can realize an OM-SKA protocol. The communication cost is ...

% }

\begin{table}[]
    \centering
\def\arraystretch{1.2}

\begin{tabular}{|c||c|c|c|c|}
	\hline 
	Protocol  & HR05 & RRS13 \& CBA15 & $\mathbf{\Pi}_{\mathrm{SKA}}$ \\ 
	\hline\hline 
	Coding Method & linear codes & polar codes  &  universal hashes \\ 
	\hline 
	Comm Cost & $\mc O(n^2)$ & $\mc O(n)$  & $\mc O(n)$ \\ 
	\hline 
	Comp A & $\mc O(n^2)$ & $\mc O(n\log n) ^{(**)}$  & $\mc O(n\log\log n)$ \\ 
	\hline 
	Comp B & $\mc O(n^2) ^{(*)}$ & $\mc O(n\log n) ^{(**)}$  & $\mc O(2^{n.H(X|Y)})$ \\ 
	\hline 
	FL Bound & yes & no  & yes  \\ 
	\hline 
\end{tabular} 
    \caption{The comparison of Protocol 1 with other protocols. Comm cost refers to the public communication bits required for the protocol. Comp A and  Comp B refer to the computation complexity of implementation for Alice and Bob, respectively. FL bound stands for finite-length lower (achievablity) bounds. (*) For the lower bound given in \eqref{eq:LB-HR-1} this computation complexity is in $\mc O(2^n)$. (**) Note that for SKAs using polar codes, parties need to compute the index sets that are required for code construction. The exact computation complexity of this step should be considered in addition to the implementation costs. }
    \label{tab:comp}
\end{table}
% \vspace{0.5 em}

% \noindent
% \textbf{Conclusion.} 
\section{Conclusion}\label{sec:Concluding remarks} %lusion}

%In this work, w
We studied  OM-SKA protocols in source model. These protocols are important for practical reasons as they do not require any interaction. They are also important from theoretical view point as they can achieve the secret key capacity of one-way SKA, and also are related to problems in computational cryptography.
Our construction uses a reconciliation method that is inspired by information spectrum analysis of \cite{Hayashi2016}.  Interesting open questions are providing efficient decoding algorithm for Bob, and   refine the reconciliation  to improve the lower bound.  
%We proposed a protocol that asymptotically achieves OW-SK capacity  and derived accurate achievable finite-length bounds for the protocol. Finally we compared the protocol with other OM-SKA protocols  in terms of computation and communication efficiency.
Obtaining %We note that finding 
finite blocklength bounds for %approximations for 
SKA using % with 
polar coding is also an % techniques, is an
 interesting open question 
for future work. %, since the asymptotic computational complexity of these protocols are practically efficient. 

\vspace{1mm}
\noindent
{\bf Acknowledgement: }This research is in part supported by Natural Sciences and Engineering Research Council of Canada, Discovery Grant program.

\bibliographystyle{IEEEtran}
%\bibliography{1.bib}
\bibliography{11.bib}

% Generated by IEEEtran.bst, version: 1.14 (2015/08/26)
\begin{thebibliography}{10}
\providecommand{\url}[1]{#1}
\csname url@samestyle\endcsname
\providecommand{\newblock}{\relax}
\providecommand{\bibinfo}[2]{#2}
\providecommand{\BIBentrySTDinterwordspacing}{\spaceskip=0pt\relax}
\providecommand{\BIBentryALTinterwordstretchfactor}{4}
\providecommand{\BIBentryALTinterwordspacing}{\spaceskip=\fontdimen2\font plus
\BIBentryALTinterwordstretchfactor\fontdimen3\font minus
  \fontdimen4\font\relax}
\providecommand{\BIBforeignlanguage}[2]{{%
\expandafter\ifx\csname l@#1\endcsname\relax
\typeout{** WARNING: IEEEtran.bst: No hyphenation pattern has been}%
\typeout{** loaded for the language `#1'. Using the pattern for}%
\typeout{** the default language instead.}%
\else
\language=\csname l@#1\endcsname
\fi
#2}}
\providecommand{\BIBdecl}{\relax}
\BIBdecl

\bibitem{Wyner1975}
A.~D. Wyner, ``{The Wire-Tap Channel},'' \emph{Bell Syst. Tech. J.}, vol.~54,
  no.~8, pp. 1355--1387, Oct 1975.

\bibitem{Ahlswede1993}
R.~Ahlswede and I.~Csisz{\'{a}}r, ``{Common randomness in information theory
  and cryptography. I. Secret sharing},'' \emph{IEEE Trans. Inf. Theory},
  vol.~39, no.~4, pp. 1121--1132, Jul 1993.

\bibitem{Maurer1993}
U.~M. Maurer, ``{Secret key agreement by public discussion from common
  information},'' \emph{IEEE Trans. Inf. Theory}, vol.~39, no.~3, pp. 733--742,
  May 1993.

\bibitem{Holenstein2006}
T.~Holenstein and R.~Renner, ``{One-Way Secret-Key Agreement and Applications
  to Circuit Polarization and Immunization of Public-Key Encryption},'' in
  \emph{Lecture Notes in Computer Science (including subseries Lecture Notes in
  Artificial Intelligence and Lecture Notes in Bioinformatics)}, 2005, pp.
  478--493.

\bibitem{holenstein2006strengthening}
T.~Holenstein, ``Strengthening key agreement using hard-core sets,'' Ph.D.
  dissertation, ETH Zurich, 2006.

\bibitem{Renes2013}
J.~M. Renes, R.~Renner, and D.~Sutter, ``{Efficient One-Way Secret-Key
  Agreement and Private Channel Coding via Polarization},'' in \emph{Lecture
  Notes in Computer Science (including subseries Lecture Notes in Artificial
  Intelligence and Lecture Notes in Bioinformatics)}, 2013, vol. 8269 LNCS, pp.
  194--213.

\bibitem{Chou2015a}
\BIBentryALTinterwordspacing
R.~A. Chou, M.~R. Bloch, and E.~Abbe, ``Polar coding for secret-key
  generation,'' \emph{IEEE Transactions on Information Theory}, vol.~61,
  no.~11, pp. 6213--6237, Nov. 2015. [Online]. Available:
  \url{http://ieeexplore.ieee.org/document/7217814/}
\BIBentrySTDinterwordspacing

\bibitem{wullschleger2007oblivious}
J.~Wullschleger, ``Oblivious-transfer amplification,'' in \emph{Annual
  International Conference on the Theory and Applications of Cryptographic
  Techniques}.\hskip 1em plus 0.5em minus 0.4em\relax Springer, 2007, pp.
  555--572.

\bibitem{Hayashi2016}
M.~Hayashi, H.~Tyagi, and S.~Watanabe, ``{Secret Key Agreement: General
  Capacity and Second-Order Asymptotics},'' \emph{IEEE Trans. Inf. Theory},
  vol.~62, no.~7, pp. 3796--3810, Jul 2016.

\bibitem{Renner2005}
R.~Renner and S.~Wolf, ``{Simple and Tight Bounds for Information
  Reconciliation and Privacy Amplification},'' in \emph{11th Int. Conf. Theory
  Appl. Cryptol. Inf. Secur. - Adv. Cryptol. - ASIACRYPT 2005}, B.~Roy,
  Ed.\hskip 1em plus 0.5em minus 0.4em\relax Chennai, India: Springer Berlin
  Heidelberg, 2005, pp. 199--216.

\bibitem{Dodis2008}
Y.~Dodis, R.~Ostrovsky, L.~Reyzin, and A.~Smith, ``{Fuzzy Extractors: How to
  Generate Strong Keys from Biometrics and Other Noisy Data},'' \emph{SIAM J.
  Comput.}, vol.~38, no.~1, pp. 97--139, Jan 2008.

\bibitem{Renner2004}
R.~Renner and S.~Wolf, ``{Smooth renyi entropy and applications},'' in
  \emph{Int. Symp. Inf. Theory, 2004. ISIT 2004. Proceedings.}\hskip 1em plus
  0.5em minus 0.4em\relax IEEE, 2004, pp. 232--232.

\bibitem{berry1941accuracy}
A.~C. Berry, ``The accuracy of the gaussian approximation to the sum of
  independent variates,'' \emph{Transactions of the american mathematical
  society}, vol.~49, no.~1, pp. 122--136, 1941.

\bibitem{esseen1945fourier}
C.-G. Esseen \emph{et~al.}, ``Fourier analysis of distribution functions. a
  mathematical study of the laplace-gaussian law,'' \emph{Acta mathematica},
  vol.~77, pp. 1--125, 1945.

\bibitem{Carter1977a}
J.~L. Carter and M.~N. Wegman, ``{Universal classes of hash functions (Extended
  Abstract)},'' in \emph{Proc. ninth Annu. ACM Symp. Theory Comput. - STOC
  '77}.\hskip 1em plus 0.5em minus 0.4em\relax New York, New York, USA: ACM
  Press, 1977, pp. 106--112.

\bibitem{Impagliazzo1989a}
R.~Impagliazzo, L.~A. Levin, and M.~Luby, ``{Pseudo-random generation from
  one-way functions},'' in \emph{Proc. twenty-first Annu. ACM Symp. Theory
  Comput. - STOC '89}.\hskip 1em plus 0.5em minus 0.4em\relax New York, New
  York, USA: ACM Press, 1989, pp. 12--24.

\bibitem{helleseth1993secret}
T.~Helleseth, ``Secret-key reconciliation by public discussion,'' in
  \emph{Advances in Cryptology̵Eurocrypt'93 (Lecture Notes in Computer
  Science)}.\hskip 1em plus 0.5em minus 0.4em\relax Springer-Verlag, 1993, pp.
  411--423.

\bibitem{Tyagi2014}
H.~Tyagi and S.~Watanabe, ``{A Bound for Multiparty Secret Key Agreement and
  Implications for a Problem of Secure Computing},'' in \emph{Adv. Cryptol. --
  EUROCRYPT 2014}, ser. Lecture Notes in Computer Science, P.~Q. Nguyen and
  E.~Oswald, Eds.\hskip 1em plus 0.5em minus 0.4em\relax Berlin, Heidelberg:
  Springer Berlin Heidelberg, 2014, vol. 8441, pp. 369--386.

\bibitem{Csiszar1978}
I.~Csisz{\'{a}}r and J.~K{\"{o}}rner, ``{Broadcast channels with confidential
  messages},'' \emph{IEEE Trans. Inf. Theory}, vol.~24, no.~3, pp. 339--348,
  May 1978.

\bibitem{Holenstein2011}
T.~Holenstein and R.~Renner, ``{On the Randomness of Independent
  Experiments},'' \emph{IEEE Trans. Inf. Theory}, vol.~57, no.~4, pp.
  1865--1871, Apr 2011.

\bibitem{Hayashi2019}
\BIBentryALTinterwordspacing
M.~Hayashi, ``{Semi-Finite Length Analysis for Secure Random Number
  Generation},'' in \emph{2019 IEEE International Symposium on Information
  Theory (ISIT)}.\hskip 1em plus 0.5em minus 0.4em\relax IEEE, Jul 2019, pp.
  952--956. [Online]. Available:
  \url{https://ieeexplore.ieee.org/document/8849241/}
\BIBentrySTDinterwordspacing

\end{thebibliography}

%%%%%%
%% To balance the columns at the last page of the paper use this
%% command:
%%
%\enlargethispage{-1.2cm} 
%%
%% If the balancing should occur in the middle of the references, use
%% the following trigger:
%%
% \IEEEtriggeratref{4}
%%
%% which triggers a \newpage (i.e., new column) just before the given
%% reference number. Note that you need to adapt this if you modify
%% the paper.  The "triggered" command can be changed if desired:
%%
\IEEEtriggercmd{\enlargethispage{-10cm}}
%%
%%%%%%

\subsection{LHL for average smooth min-entropy}

\begin{lemma}\label{smooth-hash} 
Let family $\{h_s|s\in\mathcal{S}\}$ of functions $h_s:\{0,1\}^n\rightarrow\{0,1\}^\ell$ be a 2-UHF.  Then for possibly correlated random variables $X\in\{0,1\}^n$, $V\in\{0,1\}^t$ and $Z\in\{0,1\}^*$,
$$\sd((S,Z,V,h_S(X));(S,Z,V,U_\ell))\leq 2\eps+\frac{1}{2}\sqrt{2^{\ell+t-\tilde{H}_\infty^\eps({X}|Z)}}.$$
\end{lemma}
\begin{proof}
From the average-case version of LHL in \cite[Lemma 2.3]{Dodis2008} we have
\begin{align*}
    \sd((S,Z,V,h_S(X));(S,Z,V,U_\ell))\leq 2\eps+\frac{1}{2}\sqrt{2^{\ell-\tilde{H}_\infty({X}|Z,V)}},
\end{align*}
and then for $V\in \{0,1\}^t$ from \cite[Lemma 2.2(b)]{Dodis2008} we have
$$\tilde{H}_\infty({X}|Z,V)\geq \tilde{H}_\infty({X}|Z)-t,$$ 
and therefore 
\begin{align*}
    \sd((S,Z,V,h_S(X));(S,Z,V,U_\ell))\leq 2\eps+\frac{1}{2}\sqrt{2^{\ell+t-\tilde{H}_\infty({X}|Z)}}.
\end{align*}
Since \small{$\tilde{H}_\infty^\eps(X|Z)=\displaystyle\max_{(\hat{X},\hat{Z}):\sd ((X,Y);(\hat{X},\hat{Z}))\leq \eps}\tilde{H}_\infty(\hat{X}|\hat{Z})$}, \normalsize we have $\sd(Z;\hat{Z})\leq \epsilon$ and 
$\sd(X;\Hat{X})\leq {\eps}$ and therefore, $\sd(h_S(X);h_S(\Hat{X}))\leq {\eps}$. 
We have
\begin{align*}\label{smooth-hash}
&\sd((S,Z,h_S(X));(S,Z,U_\ell))\\
&\qquad \leq \sd((S,Z,h_S(X));(S,Z,h_S(\hat{X}))) \\
&\qquad\qquad+\sd((S,Z,h_S(\Hat{X}));(S,\Hat{Z},h_S(\Hat{X}))\\
&\qquad\qquad+\sd((S,\Hat{Z},h_S(\Hat{X}));(S,Z,U_\ell))\\
&\qquad\leq 2\eps+\frac{1}{2}\sqrt{2^{\ell-\tilde{H}_\infty(\Hat{X}|\Hat{Z})}}.\qedhere
% \\\label{smooth-hash}
% &\qquad\qquad\leq \eps+\frac{1}{2}\sqrt{2^{\ell-\tilde{H}_\infty^\eps({X}|Z)}}.
\end{align*}\end{proof}

\end{document}